\documentclass{llncs}
\pdfoutput=1 
\usepackage{llncsdoc}
\usepackage[english]{babel}
\usepackage[utf8]{inputenc}
\usepackage[ngerman,colorlinks=true]{hyperref}
\usepackage{amsmath,amssymb,latexsym}
\usepackage{wasysym,stmaryrd}
\usepackage{tikz}
\usepackage{float}
\usepackage{wrapfig}
\usepackage{algpseudocode}
\usepackage{algorithm} 
\usepackage{environ}
\usepackage{graphicx}
\usepackage{subfigure}
\usepackage{listings}
\usepackage{color}
\usepackage[all]{xy}
\input{xy} \xyoption{all}

\newcommand{\reconnet}{\textsc{ReCon}\textsc{Net}}

\newcommand{\Pot}{\mathcal{P}}
\newcommand{\R}{\mathcal{R}}

\newcommand{\deriv}[1]{\stackrel{#1}{\Longrightarrow}}

\newcommand{\nach}[1]{\stackrel{#1}{\to}}
\newcommand{\von}[1]{\stackrel{#1}{\gets}}

\newcommand{\N}{{\mathbb{N}}}

\newcommand{\M}{\ensuremath{\mathcal{M}}}

\newcommand{\categ}[1]{\ensuremath{\mathbf{ #1}}}
\newcommand{\cat}{\categ{C}}
\newcommand{\cC}{\categ{C}}
\newcommand{\cSets}{\categ{Sets}}
\newcommand{\cPT}{\categ{PT}}
\newcommand{\cdPT}{\categ{decoPT}}

\newcommand{\cdPTi}{\categ{decoPTi}}
\newcommand{\cdPTip}{\categ{decoPTip}}
\newcommand{\cPosets}{\categ{PoSets}}
\newcommand{\cPTp}{\categ{PTp}}
   
\newcommand{\fire}[1]{[{#1}\rangle}

\newcommand{\oder}{\vee}
\newcommand{\und}{\wedge}
{\bfseries}{\itshape}

\begin{document}
\title{Reconfigurable Decorated PT Nets with Inhibitor Arcs and Transition Priorities}% and other Control Structures}
\author{Julia Padberg}
              
\institute{Hamburg University of Applied Sciences\\Germany}

\maketitle
\begin{abstract}
In this paper we deal with additional control structures for decorated PT Nets. The main contribution are inhibitor arcs and  priorities. The first ensure that a marking can inhibit the firing of a transition.
Inhibitor arcs force that the transition may only fire when the place is empty.
an order of transitions restrict the firing, so that an transition may fire only if it has the highest priority of all enabled transitions.
This concept is shown to be compatible with reconfigurable Petri nets.

\end{abstract}


\section{Introduction}
%intro
Motivation for reconfigurable Petri nets, a  family of formal modelling techniques (e.g. in \cite{EP03,LO04,EHP+07,PEHP08,KCD10}) is  the observation that  
		in increasingly many application areas the underlying system has to be dynamic in a structural sense. 
		Complex coordination and structural adaptation at run-time  (e.g. mobile ad-hoc networks,
		communication spaces, ubiquitous	computing) are main features that need 
		to be modelled adequately.  The distinction between the net behaviour and the dynamic change of its net structure is 		the characteristic feature that makes reconfigurable Petri nets so suitable for systems with dynamic structures. 
	
Reconfigurable Petri  nets  consist of marked Petri nets, i.e. a net with a marking,
and a set of rules whose application modifies the net's structure at runtime. Typical application areas are  concerned with the modelling of  dynamic structures, for example  workflows in a dynamic infrastructure. 

As  an abstract example of a dynamic system  we use a cyclic process that can either be executed or  modified .
These modifications  change the process by inserting additional sequential steps or by forking into parallel steps and they can be reversed too.  The  net in Fig. \ref{n.start} describes a cyclic process with a distinguished  place \texttt{start} that can execute one step and then returns to the start. The modifications
are modelled by the rules given in Fig. \ref{f.rules}. 
\begin{figure}[h]
		\centering
     \subfigure[rule \texttt{sequential\_ext\_s} \label{r.seq_ext_s}]{\includegraphics[width=8cm]{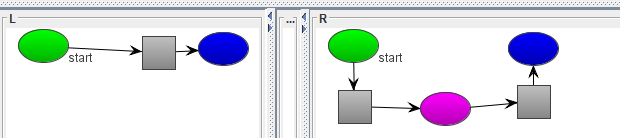}}
		
     \subfigure[rule \texttt{sequential\_ext} \label{r.seq_ext}]{\includegraphics[height=12mm]{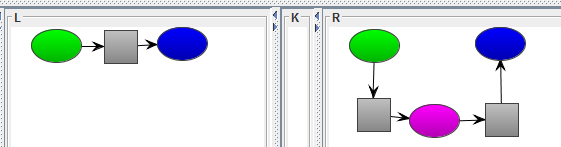}}
     \subfigure[rule \texttt{sequential\_red} \label{r.seq_red}]{\includegraphics[height=12mm]{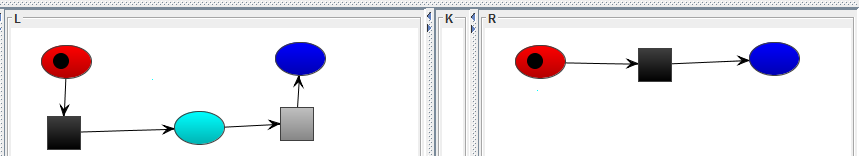}}

     \subfigure[rule \texttt{parallel\_ext} \label{r.par_ext}]{\includegraphics[height=17mm]{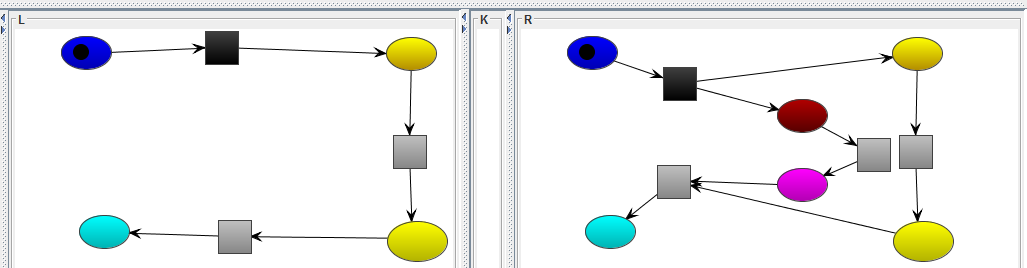}}
     \subfigure[rule \texttt{parallel\_red} \label{r.par_red}]{\includegraphics[height=17mm]{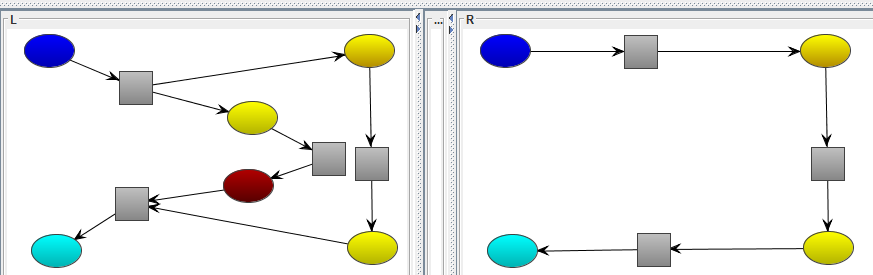}}
		\label{f.rules}
		\caption{Rules}
\end{figure}
The colours of the places and transitions indicate the mappings within  the rule.  Rule \texttt{sequential\_ext\_s} in Fig. \ref{r.seq_ext_s}  models the first possible modification, the insertion of a sequential step after the  place \texttt{start}. The left-hand side of the rule is the net $L$ and shows the places that need to be in the context and the transition that is deleted. In the right hand side of the rule is the net $L$ and shows the added place and transitions as well as the context. For reasons of space we have omitted the intermediate net $K$ that denotes the context explicitly.
the rule \texttt{sequential\_ext\_s}
is the first rule that can be applied by matching the  place \texttt{start} in $L$ to the  place \texttt{start} in net \texttt{start\_net} in Fig. \ref{n.start}. The application of a rule via a match from $L$ to the given net leads then to the direct transformation 
from the given net to  the resulting net and is achieved by deleting and adding according to rule.

Reconfigurable Petri nets allow the application of these rules together with the firing of the transitions.
Let the application of  rule \texttt{sequential\_ext\_s} be the first step, followed by a firing step. This results in the net  in Fig.~\ref{n.2steps}. The resulting net has an additional place and an additional transition, denoting the process to have been modified by inserting a sequential step. Moreover, the next step has already been executed denoted by firing the transition in the post-domain of place \texttt{start}. 
\begin{figure}[h]
   \parbox[b]{4.5cm}{
     \subfigure[net \texttt{start\_net} \label{n.start}]{\includegraphics[width=3cm]{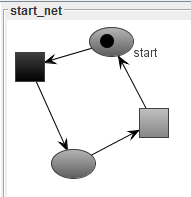}}\\
     \subfigure[net  after 2 steps\label{n.2steps}]{\includegraphics[width=4.5cm]{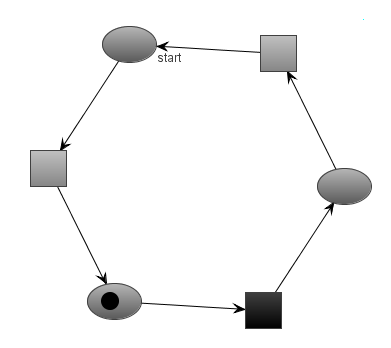}}
		}
     \subfigure[net after 10 steps\label{n.10steps}]{\includegraphics[width=7.5cm]{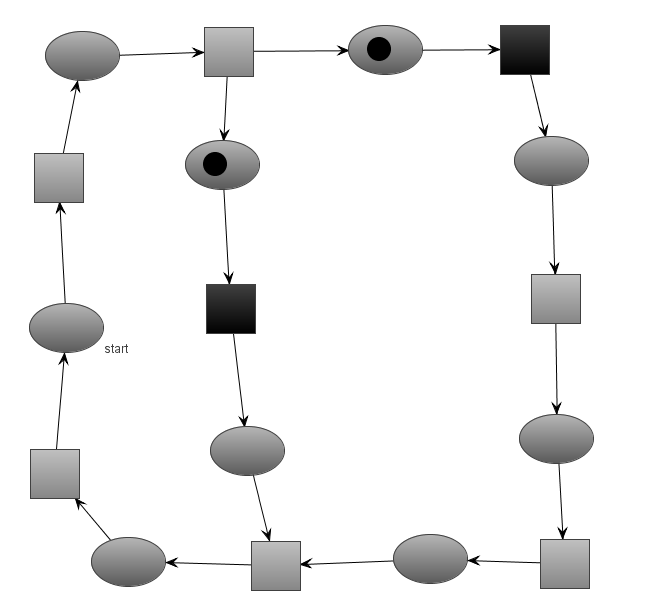}}\\
	\caption{Start and intermediate nets \label{f.trafo}}
\end{figure}
These steps are chosen non-deterministic so the start net in Fig. \ref{n.start} may evolve in ten steps to the net  in Fig. \ref{n.10steps} by  firing transitions or applying rules.  Due to the application of rule \ref{r.par_ext} we now have a fork and due to the firing of the forking transition we have two token.
After another 20 steps it may look like the net in Fig. \ref{n.20steps}. Note, that the rules \ref{r.seq_ext} and  \ref{r.seq_red} are inverse to each other as well as he rules \ref{r.par_ext} and  \ref{r.par_red}. So, after another 20 steps the net  may as well be back to the net in Fig.~\ref{n.2steps}, but it cannot reach the start net as there is no inverse rule to rule \ref{r.seq_ext_s}.

For the sake of the main focus we have considered merely a small and abstract example. More complex nets and rules can be found in case studies for the applications of reconfigurable Petri nets, see e.g. \cite{Rei12,MH10,HEP08}.

\begin{figure}[H]
\centering
   \includegraphics[width=10cm]{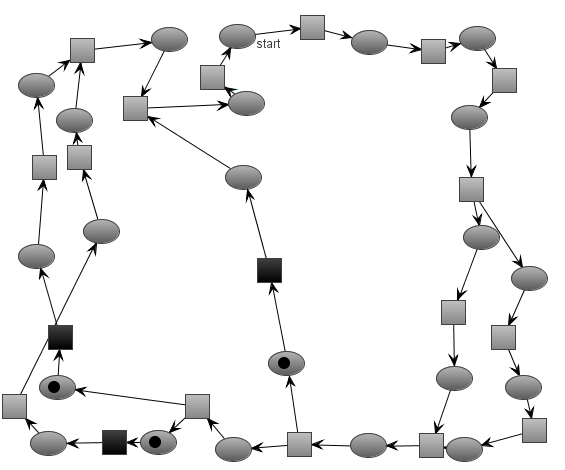}
	\caption{Net after another 20 steps \label{n.20steps}}
\end{figure}

The paper is organized as follows: First we introduce  decorated place/transition nets adding some annotations as names and renewable labels. We motivate changing transition labels and extend the firing of a transition so that the labels may be changed. Nevertheless, this extension is conservative to the firing behavior. Then we define reconfigurable Petri nets  based on decorated place/transition  nets. In the next section  we add inhibitor arcs  to decorated PT nets and show, that they are still $\M$-adhesive. Section  \ref{s.prior} extends the set of transitions with a partial order, describing the priorities between the transitions. We employ the category of partial orders $\cPosets$ and again we obtain an $\M$-adhesive category.

\section{Reconfigurable Petri Nets}
\label{s.recPN}
%reconPN

We use the algebraic approach to Petri nets, so a marked place/transition net is given by
 $N=(P,T,pre,post,M)$ with pre- and post-domain functions $pre,post: T \to P^\oplus$ and a  marking $M \in P^\oplus$, 
 where $P^\oplus$ is the free commutative monoid over the set $P$ of places.  To obtain the weight of an arc from a place to a transition $t$ the pre domain function is restricted to that place, i.e. $pre(t)_{|p} \in \N$; analogously the weight of an  arc from a transition to a place is given by the restriction of the post domain function. 
 For $M_1, M_2 \in P^\oplus$ we have $M_1 \leq M_2$ if $M_1(p) \leq M_2(p)$ for all $p \in P$.  A transition $t \in T$ is $M$-enabled for a marking
$M \in P^\oplus$ if we have $pre(t) \leq M$, and in this case the follower
marking $M^\prime$ is given by $M^\prime=M \ominus pre(t) \oplus post(t)$ and $M\fire{t} M^\prime$ is called firing step. In \cite{Pad12} new features have been added to gain an adequate modelling technique.
The extension to  capacities and names is quite obvious. More interesting are the transition labels that may change, when the transition is fired. This allows a better coordination of transition firing and rule application, for example can be ensured that a transition has fired (repeatedly) before a transformation may take place. This last extension is 
 conservative with respect to Petri nets as it does not change the net behaviour. 

\subsection{Decorated Place/Transition Nets}
 
A decorated place/transition  net is a marked P/T net $N=(P,T,pre,post,M)$ together with names and labels. A capacity is merely a function $cap :P\to \N^\omega_+$. Based on name spaces $A_P$, $A_T$   with $pname:P \to A_P$ and $tname:T \to A_T$ we have explicit names for places and transitions.  Moreover,  transitions are equipped with labels that may change when the transition fires. This 
feature is given by a  mapping of transitions to functions.
For example the net $N_2$ in Fig. \ref{f.pnvf2_ex1} yields the marking $3p_a+p_b+2p_c$ 
after firing transitions $t_b$ and $t_d$ in parallel. Furthermore, this parallel firing yields the new transition labels $2$ for transition $t_b$ and $false$ for transition $t_d$. So, we compute the follower label $tlb \fire{t_b+t_d} tlb'$, where $tlb, tlb': T\to W$ are label functions with
$tlb'(t_b) = inc(tlb(t_b)) = inc(1) = 2$, where the renew function $inc:\N\to \N$ increases the label by one and
$tlb'(t_d) = not(tlb(t_d)) = not(true) = false$. For more details see \cite{Pad12}.
		
\begin{definition}[Decorated place/transition  net]   
A decorated place/transition  net is a marked place/transition net $N=(P,T,pre,post,M)$ together with
			\begin{itemize}
				\item a capacity as a function $cap :P\to \N^\omega_+$ 
				\item   name spaces $A_P$, $A_T$ with 
				      $pname:P \to A_P$ and $tname:T \to A_T$
			  \item the function $tlb: T \to W$ mapping transitions to transition labels $W$ and
			  \item the function
		          $rnw: T \to END$ where $END$ is a set containing some endomorphisms on
		          $W$, so that 
		             $rnw(t): W \to W$ is the function that renews the transition label.
			\end{itemize}
\end{definition}

The firing of these nets is the usual for  place/transition nets except for changing the transition labels.
 Moreover, this extension works for parallel firing as well.
   
\begin{definition}[Changing Labels by Parallel Firing]
   Given a transitions vector $v= \sum_{t \in T} k_t \cdot t$
   then the label is renewed by firing  $tlb \fire{v} tlb'$ and for each $t \in T$ the transition label 
   $tlb':T \to W$ is defined by:
    $$tlb'(t) = rnw(t)^{k_t} \circ tlb (t)$$
\end{definition}

\subsection{Transformations of Decorated Nets}
For decorated  place/transition nets as given above, we obtain with the following notion of morphisms an $\M$-adhesive HLR category (see \cite{Pad12}).  
$\M$-adhesive HLR systems can be considered as a unifying framework for graph and Petri net transformations
 providing enough structure
that most notions and results from algebraic graph transformation systems are available, as results on parallelism and concurrency of rules and transformations, results on negative application conditions and constraints, and so on (e.g. in \cite{FAGT,EGHLO12}).

Net morphisms map places to places and transitions to transitions. They are given as a pair of mappings for the places and the transitions, so that the structure and the decoration is preserved and the marking may be  mapped strictly.

\begin{definition}[Morphisms between  decorated place/transition nets \cite{Pad12}]\label{def.morphism} 
A net morphism $f:N_1 \to N_2$ between two decorated place/transition  nets $N_i =( P_i,T_i, pre_i, post_i,M_i,cap_i, pname_i,tname_i, tlb_i, rnw_i)$ for $i\in\{1,2\}$ is given by $f=(f_P:P_1 \to P_2,f_T:T_1 \to T_2)$, so that the following equations hold: 
\begin{enumerate}
	\item \label{i}   $pre_2 \circ f_T = f_P^\oplus \circ pre_1$ and $post_2 \circ f_T = f_P^\oplus \circ post_1$ 
	\item \label{ii}  $cap_1 = cap_2 \circ f_p$ 
	\item \label{iii} $pname_1 = pname_2 \circ f_P$ 
	\item \label{iv}  $tname_1 = tname_2 \circ f_T$ and  $tlb_1 = tlb_2 \circ f_T$ and $rnw_1 = rnw_2 \circ f_T$
	\item \label{v} $M_1(p) \le M_2(f_P(p))$ for all $p \in P_1$
\end{enumerate}
Moreover, the morphism $f$ is called strict 
\begin{enumerate}
   \setcounter{enumi}{5}
	\item \label{vi} if  both  $f_P$ and $f_T$  are injective and 
 $ M_1(p) =M_2(f_P(p))$  holds for all $p \in P_1$.
\end{enumerate}
\end{definition}

A rule in the DPO approach  is given by three nets
called left hand side $L$, interface $K$ and right hand side $R$, respectively, and a span of two strict  net morphisms $K \to L$ and $K\to R$.\\    
  Additionally, a match morphism $m: L\to N$ is required that identifies the
relevant parts of the left hand side in the given net $N$. Then a transformation step
$N \deriv{(r,m)}M$ via rule $r$ can be constructed in two steps. 
Given a rule with a match $m:L\to N$ the gluing conditions  have to be satisfied in order to apply a rule at a given match. These conditions ensure  the result is again a well-defined net.
It is a sufficient condition for the existence and uniqueness  of the  so-called pushout complement which is needed for the first step in a transformation.
In this case, we obtain a net  $M$   leading to  a direct transformation $N \deriv{(r,m)} M$ consisting of the following pushouts (1) and (2) in Fig. \ref{dpo}.              
\begin{wrapfigure}[7]{r}{.4\textwidth}
$
\xymatrix{  L  \ar[d]|m \ar@{}[dr]|{\bf (1)}   
            & K \ar[l] \ar[r] \ar[d]   \ar@{}[dr]|{\bf (2)}
            & R \ar[d]\\
             N
            & D \ar[l] \ar[r]
            & M
              }
 $
	\caption{Transformation of a net}
	\label{dpo}
\end{wrapfigure}

Next we show that 
decorated place/transition nets yield  an $\M$-adhesive  HLR category for $\M$ being the class of strict morphisms. Hence we obtain all the well-known results, as transformation, local confluence and parallelism, application conditions, amalgamation and so on.

\begin{lemma}[see \cite{Pad12}]
\label{l.mahlr}
The category $\cdPT$ of decorated place/transition nets  is an  $\M$-adhesive HLR category.
\end{lemma}

This construction as well as  a huge amount of notion and results are available since          
decorated place/transition nets can be proven to be an $\M$-adhesive HLR category. 
Hence we can combine one net together with a set of rules
leading to reconfigurable place/transition nets. 

\begin{definition}[Reconfigurable Nets]
    A reconfigurable decorated place/transition  net $RN=(N,\R)$ is given by an decorated $N$
    and a set of rules $\R$.
\end{definition}

\section{Review of $\M$ adhesive HLR Systems}
\label{s.madHLR}
% madHLR.tex

The theory of HLR systems has been developed as an abstract framework for 
different types of graph and Petri net 
transformation systems. Moreover the HLR framework has been applied to 
algebraic specifications \cite{pEM90}, where the interface of an algebraic 
module specification can be considered as a production of an algebraic 
specification transformation system \cite{pEGP99}.  HLR systems are instantiated with various types of 
graphs, as hypergraphs, attributed and typed graphs,  structures, algebraic specifications, 
various Petri net classes, elementary nets, place/transition nets, Colored Petri nets, 
or algebraic high-level nets, and more (see \cite{gEHKP91} and \cite{EEPT05}).
Adhesive categories have been introduced in \cite{LS04} and have been 
combined with HLR categories and systems in \cite{gEHPP04} leading to the new 
concept of (weak) adhesive HLR categories and systems. 
The main reason why adhesive categories are important for the theory of graph 
transformation and its generalization to high-level replacement systems
is the fact that most of the HLR conditions required in \cite{gEHKP91} are shown to 
be already valid in adhesive categories (see \cite{LS04}).
The fundamental construct for (weak) adhesive (HLR) categories and 
systems are  van Kampen (VK) squares.

\begin{definition}[\M-Van Kampen square]\label{d.VK}
A pushout (1) with $m \in \M$  is a \M-van Kampen (VK) square, if for any commutative cube (2) with (1) in 
the bottom and back faces being pullbacks, the following holds:\\
 the top is pushout $\Leftrightarrow$ the
front faces are pullbacks.

  $
   \xymatrix@=3mm{
   			&\\
                       A  \ar[rr]|{m\in\M} \ar[dd]|{f}    
	         && B                   \ar[dd]|{g}     \\ 
	            & (1)  \\
		       C  \ar[rr]|{n}
		    && D
		 }
     $ \hfill
     $
    \xymatrix@=3mm{
                     &&&    A'	 \ar[ddd]|<<<<<{a}    \ar[dlll]|{f'}    \ar[drr]|{m'}     
		         &&  (2) \\
		                C'      \ar[ddd]|{c}                    \ar[drr]|{n'}
		   &&&&& B'      \ar[ddd]|{b}     \ar[dlll]|<<<<<{g'}          
		           &&                                                                \\
		          && D'      \ar[ddd]|<<<<<{d}                                        \\
                        &&& A                       \ar[dlll]|>>>>>{f}     \ar[drr]|{m}           \\
		                C                                      \ar[drr]|{n}
		   &&&&& B                       \ar[dlll]|{g}                       
		           &&                                                              \\
		          && D
		   } 	        	 
  $	
\end{definition}

\M-adhesive HLR systems can be considered as abstract 
 transformation systems in the double pushout approach based on \M-adhesive 
HLR categories.

\begin{definition}[$\M$-Adhesive HLR Category and PO-PB Compatibility \cite{EGH10}]
\label{d.madHLR}
 Given a PO-PB
compatible class $\M$ of monomorphisms in $\cC$ (see below), then $(\cC,\M)$ is called
$\M$-adhesive HLR-category, if pushouts along \M-morphisms are $\M$-VK squares (see \ref{d.VK}).\\

A class $\M$ of monomorphisms in $\cC$ is called PO-PB compatible, if
\begin{enumerate}
	\item Pushouts along \M-morphisms exist and \M \ is stable under pushouts.
  \item Pullbacks along \M-morphisms exist and \M \ is stable under pullbacks.
  \item \M \ contains all identities and is closed under composition.
\end{enumerate}

 An \M-adhesive 
HLR system $AHS = (\cat,\M,P)$  consists of an adhesive HLR category $(\cat,\M)$  and a set of rules $P$. 
\end{definition}

\section{Inhibitor Arcs}
\label{s.inhib}
% inhib.tex

We here introduce generalized inhibitor arcs, that may consider several places to inhibit the transitions firing. So inhibitor arcs are given as a function fro transitions to the multiset of places.

\begin{definition}[Generalized inhibitor arcs]
    Given a  decorated place/transition  net $N =(P,T, pre, post,M,cap, pname,tname, tlb, rnw)$
		inhibitor arcs are given by $inh:T \to \Pot(P)$.
		
		A transition is then enabled  under a marking $M_1$ if additionally  we have $M_1(p)=0$ for all $p\in inh(t) $.
\end{definition}

\begin{lemma}
\label{l.mahlr}
The category $\cdPTi$ of decorated place/transition nets with inhibitor arcs  is an  $\M$-adhesive HLR category with $\M$ being the class of strict, injective net morphisms.
\end{lemma}
\begin{proof}
The proof applies the construction for weak adhesive HLR categories (see Theorem 1 in \cite{PEL08}):\\
Constructing the category $\cdPTi$ using comma categories,  we use  the functor $F: \cdPT \to \cSets$ yielding the transition set $T$  and the power set functor $\Pot : \cSets \to \cSets$.
    The category of decorated place/transition  nets is a $\M$-adhesive HLR category (see \cite{Pad12}):
    Then the comma category $\categ{\cdPTi}:= CommCat(F, \Pot, \{inh\}$) yields 
    the category of  decorated place/transition  nets with inhibitor arcs and is a weak adhesive HLR category as $F$ preserves pushouts and $\Pot$ pullbacks of injective morphisms.
		
		Hence, we have an  $\M$-adhesive HLR category, see \cite{EGH10}.
     \end{proof}

\section{Transition Priorities}
\label{s.prior}
 % prior
The set of transitions $T$ is equipped with a partial order $\le$ on the transitions. $t$ is enabled under a marking $M$, if $pre(t) \ge M$, if $cap(t) \ge M+post(t)$ and if 
 all $t’$ being enabled under $M$ we have $t’\le t$.

We first need to investigate the category \cPosets of partially ordered sets. In \cite{Cod07} this category has been examined.

\begin{definition}[Category \cPosets]
   The objects are partially orders sets, given by a set $P$ and a partial order $\le$ over $P$.
	 The morphisms if this category are order-preserving maps, that are maps $f:P_1 \to P_2$  preserving the order, so $x\le y$ implies $f(x) \le f(y)$.
\end{definition}
Composition  and identity are defined as for sets and are both order-preserving,  \cPosets is indeed a category \cite{Cod07}.

The relation to the category of sets can be given by two functors.
The free functor  $F:\cSets \to \cPosets$ is given by $F(M\nach{f}M')= (M,ID_M) \nach{f} /nach{f} (M',ID_{M'}$ where $ID_M$ is the identity relation of a set $M$.
The forgetful functor $V: \cPosets \to \cSets$ is defined by $V( (P,\le_P) \nach{g} (P',\le_{P'})) = P \nach{g} P'$.

\begin{lemma}[Adjunction to \cSets]

\end{lemma}

\begin{proof}
\end{proof}
So, we know that $F$ preserves  colimits ans $V$ preseves limits.

\begin{lemma}[Initial Object and Pushouts in \cPosets]
\label{l.poPosets}
\begin{enumerate}
	\item The initial object is $(\emptyset,\emptyset)$.
	\item Given the span $(P_1,\le_1) \von{f} (P_0,\le_0) \nach{g} (P_2,\le_2)$, then there exists  the pushout
	     $(P_1,\le_1) \nach{g'} (P_3,\le_3) \von{f'} (P_2,\le_2)$.
\end{enumerate}
\end{lemma}

\begin{proof}
\begin{enumerate}
	\item The initial object is $(\emptyset,\emptyset)$ as there is the empty order preserving 
	mapping to each partially orderes set in \cPosets.
	
	\item Given $(P_1,\le_1) \von{f} (P_0,\le_0) \nach{g} (P_2,\le_2)$, then 
	there is in $\cSets$ the span\\
	       $P_1\von{f} P_0 \nach{g} P_2$ and its pushout $P_1 \nach{\bar{g}} \bar{P_3} \von{\bar{f}} P_2$, see pushout $(PO)$ in \ref{eq.SetsPO}				
				 and the relation $R_3 \subseteq \bar{P_3} \times \bar{P_3} $ with 
				\begin{equation}
	         \label{eq.R3}
		       \begin{aligned}
					    (x_3,y_3) \in R_3 &\textrm{ if and only if } \\
					                      & \exists x_1,y_1 \in P_1: \bar{g}(x_1) =x_3 \und \bar{g}(y_1) =y_3 
																                                             \und x_1 \le_1 y_1 \\
					               \oder  & \exists x_2,y_2 \in P_2: \bar{f}(x_2) =x_3 \und \bar{f}(y_2) =y_3 
																                                             \und x_2 \le_2 y_2
				   \end{aligned}
         \end{equation}
				
				Since $R_3$ is not a partial order\footnote{%
				         Let $P_0=\{0,5\}$ and $P_1= \{0,3,5\}$ with $f$ the inclusion and $P_2=\{\bullet\}$,
								 then $3 \le_1 5$ yields $([3],[\bullet]) \in R_3$ and $0 \le_1 3$ yields $([\bullet],[3]) \in R_3$, but $[\bullet]=\{0,5\} \neq \{3\}=[3]$.
				        },
				we define the relation $\bar{R_3}$ to be the equivalence closure  of all
				symmetric pairs $\{(x_3,y_3) \mid (x_3,y_3), (y_3,x_3) \in R_3\}  \subseteq R_3$.  Then  we have the quotient
			   $P_3 =  \bar{P_3}_{\mid \bar{R_3}}$ with $g':= [ \_ ] \circ \bar{g}: P_1\to P_3$ and $f':= [ \_ ] \circ \bar{f}: P_2\to P_3$, where $[\_]: \bar{P_3} \to  \bar{P_3}_{\mid \bar{R_3}}=P_3$ is the natural function mapping each element of $\bar{P_3} $ to its equivalence class.\\
					$\le_3$  is the transitive closure of \\
								 $\{(x_3,y_3) \mid $\\ \hspace*{5mm}
											$x_1 \le_1 y_1  \textrm{ for } g'(x_1) =x_3 \textrm{ and } g'(y_1) =y_3 $\\ \hspace*{3mm}
											or\\ \hspace*{5mm}
											$
											x_2 \le_2 y_2  \textrm{ for } f'(x_2) =x_3 \textrm{ and } f'(y_2) =y_3 
								\}$
\end{enumerate}
       $\le_3$ is a partial order, as it is reflexive, antisymmetric and transitive and $f'$ and $g'$ are order-preserving maps by construction.\\
	     So, in $\cPosets$ the category of partially ordered sets  $(P_1,\le_1) \nach{g'} (P_3,\le_3) \von{f'} (P_2,\le_2)$ is the pushout of $(P_1,\le_1) \von{f} (P_0,\le_0) \nach{g} (P_2,\le_2)$:

				\begin{equation}\label{eq.SetsPO}
					$$\xymatrix{
												P_0 \ar[r]^{f} \ar[d]_{g} \ar@{}[dr]|{(PO)}
											&  P_1 \ar[d]^{\bar{g}}  \ar@/^4mm/[dr]|{g':= [ \_ ] \circ \bar{g}} \ar@/^11mm/[ddrr]^{g''}\\
								P_2 \ar[r]_{\bar{f}}  \ar@/_8mm/[rr]|{f':= [ \_ ] \circ \bar{f}} \ar@/_11mm/[drrr]_{f''}
							&  \bar{P_3} \ar[r]|{[ \_ ]}  \ar[drr]|{\bar{h}}
							&  P_3    \ar[dr]^{h} \\
						&&&  P_4
					}
					$$
\end{equation}
			Obviously $g' \circ f = f'\circ g$. \\
			For any partially ordered set $(P_4,\le_4)$ with
			$g'' \circ f= f''\circ g$ we have $\bar{h}: \bar{P_3} \to P_4$ in $\cSets$ due to the pushout $(PO)$ in Diagram \ref{eq.SetsPO}. So, we define $h: P_3 \to P_4$ with $h([x]) = \bar{h}(x)$. \\
			To prove that $h$ is well-defined we show $h([x_3]) = h([y_3])$ with $x_3 \neq y_3$ 
			but 	$[y_3] = [x3]$.
			
			Since 	$[y_3] = [x3]$ and $x_3 \neq y_3$ there is $(x_3,y_3) \in \bar{R_3}$  and hence $(x_3,y_3) \in R_3$ and $(y_3,x_3) \in R_3$. Due to the definition of $R_3$ there are four cases:
			
			\begin{enumerate}
				\item $\exists x_1,y_1 \in P_1 : x_1 \le_1 y_1 \und \bar{g}(x_1) = x_3 \und \bar{g}(y_1) = y_3 $\\
				      $\und  \; \exists x_2,y_2 \in P_2 : y_2 \le_2 x_2 \und \bar{f}(x_2) = x_3 \und \bar{f}(y_2) = y_3 $:
							
							Then we have $g''(x_1) = \bar{h} \circ \bar{g} (x_1) =  \bar{h} \circ \bar{f} (x_2) = f''(x_2) $
							and $g''(y_1) = \bar{h} \circ \bar{g} (y_1) =  \bar{h} \circ \bar{f} (y_2) = f''(y_2) $.\\
							This yields
							$g''(x_1) \le_4 g''(y_1)$  and $g''(y_1) = f''(y_2) \ge_4 f''(x_2) = g''(x_1)$. Since $\le_4$ is a antisymmetric we have $g''(x_1)= g''(y_1)$.\\
							Hence, we have $h([x_3]) = \bar{h}(x_3) = \bar{h}\circ \bar{g} (x_1) = g''(x_1)= g''(y_1) =
							\bar{h}\circ \bar{g} (y_1) = \bar{h}(y_3) = h([y_3])$. \\
					\item $\exists x_1,y_1 \in P_1 : y_1 \le_1 x_1 \und \bar{g}(x_1) = x_3 \und \bar{g}(y_1) = y_3 $\\
				      $ \und  \; \exists x_2,y_2 \in P_2 : x_2 \le_2 y_2 \und \bar{f}(x_2) = x_3 \und \bar{f}(y_2) = y_3 $ analogously.\\
					\item $\exists x_1,y_1 \in P_1 : x_1 \le_1 y_1 \und \bar{g}(x_1) = x_3 \und \bar{g}(y_1) = y_3 $\\
				      $ \und \; \exists x'_1,y'_1 \in P_1 : y'_1 \le_1 x'_1 \und \bar{g}(x'_1) = x_3 \und \bar{g}(y'_1) = y_3 $:
							
							So, we have $\bar{g}(x_1) = x_3 = \bar{g}(x'_1)$ and $\bar{g}(y_1) = y_3 = \bar{g}(y'_1)$. and 
							   $x_1 \le_1 y_1$ and $ y'_1 \le_1 x'_1$.\\
							This yields
							$g''(x_1) \le_4 g''(y_1)$  and $g''(y_1) = g''(y'_1) \le_4 g''(x'_1) = g''(x_1)$. Since $\le_4$ is a antisymmetric we have $g''(x_1)= g''(y_1)$.\\
							Hence, we have $h([x_3]) = \bar{h}(x_3) = \bar{h}\circ \bar{g} (x_1) = g''(x_1)= g''(y_1) =
							\bar{h}\circ \bar{g} (y_1) = \bar{h}(y_3) = h([y_3])$. \\
					\item $\exists x_2,y_2 \in P_2 : x_2 \le_2 y_2 \und \bar{f}(x_2) = x_3 \und \bar{f}(y_2) = y_3 $\\
				      $ \und \; \exists x'_2,y'_2 \in P_2 : y'_2 \le_2 x'_2 \und \bar{f}(x'_2) = x_3 \und \bar{f}(y'_2) = y_3 $
							analogously.
			\end{enumerate}
			Moreover, $h \circ g' = h \circ [\_] \circ  \bar{g} = \bar{h} \circ  \bar{g} = g''$ and 
			$h \circ f' = h \circ [\_] \circ  \bar{f} = \bar{h} \circ  \bar{f} = f''$.
			
\end{proof}

Next we introduce the subclass of monomorphisms $\M$. Monomorphisms in $\cPosets$ are the injective order preserving  maps \cite{Cod07} and order embeddings - those mappings that satisfy item \ref{d.M.i} in Def.~\ref{d.M} -are regular monomorphisms \cite{Cod07}.

\begin{definition}[Class $\M$]
\label{d.M}
The class $\M$ is given by the class of strict order embeddings,  that are order preserving mappings
$f:(P,\le_P) \to (P',\le_{P'})$ that additionally satisfy :

\begin{enumerate}
	\item \label{d.M.i} $x\le_P y$ if and only if $f(x) \le_{P'} f(y)$ for $x,y  \in P$ 
	\item \label{d.M.ii} for each $  z' \in P'$ with $f(x) \le_{P'} z'  \le_{P'} f(y)$  there exists some $z \in P$ with $f(z) = z'$ (and hence $x \le_P z \le_P y$).
\end{enumerate}
\end{definition}

Class $\M$ leads to pushouts that are constructed as in the category $\cSets$, hence the forgetful functor
$V:\cPosets\to \cSets$ preserves pushouts.

\begin{lemma}[$\M$-Pushouts in \cPosets]
Given $(P_1,\le_1) \von{f} (P_0,\le_0) \nach{g} (P_2,\le_2)$  with $f\in \M$ then there is the pushout
	     $(P_1,\le_1) \nach{g'} (P_3,\le_3) \von{f'} (P_2,\le_2)$,
			   such that in $\cSets$ $P_1 \nach{g'} P_3 \von{f'} P_2$  is the pushout of 
					        $P_1\von{f} P_0 \nach{g} P_2$.\\
				Moreover,  \M \ is stable under pushouts.
\end{lemma}

\begin{proof}
Obviously, the construction of $\bar{R_3}$ in the proof of Lemma \ref{l.poPosets} yields for $f \in \M$ that 
$\bar{R_3} =ID$ the identity relation. Hence , $\bar{P_3}= \bar{P_3}_{\mid \bar{R_3}}=P_3$.\\
Moreover, its is \M-stable:\\
For $f \in \M$ in Diagram \ref{eq.SetsPO} we know that $f'$ is injective, as pushouts in $\cSets$ preserve monomorphisms, i.e. injective mappings and it is order-preserving by construction. \\
$f'$ is an order embedding: \\
For $x_2, y_2 \in P_2$ and $f'(x_2) \le_3 f'(y_2)$ we have due to the construction of $\le_3$ four cases:

\begin{enumerate}
	\item There are $x_1,  y_1 \in P_1$ with $x_1 \le_1 y_1$ so that $g'(x_1) = f'(x_2)$ and $g'(y_1) = f'(y_2)$.
	      Due to the pushout construction there are  $x_0,  y_0 \in P_0$ with $x_0 \le_0 y_0$ 
				  so that $f(x_0) = x_1$ and $g(x_1) = x_2$ and  $f(y_0) = y_1$ and $g(y_1) = y_2$. Since $g$ is order preserving, we have $x_2 \le_2 y_2$.\\					
	\item There is $x_2 \le_2 y_2$.\\
	\item There is $z_3 \in P_3$ with $f'(x_2) \le_3 z_3 \le_3 f'(y_3)$, so that there are $x_1 \le_1 z_1$ with
	      $g'(x_1) = f'(x_2)$  and $g'(z_1) = z_3$ and $z_2 \le_2 y_2$ and $f'(z_2) =z_3$.\\
				Due to the pushout construction there are  $x_0,  z_0 \in P_0$ with $x_0 \le_0 z_0$ 
				  so that $f(x_0) = x_1$ and $g(x_1) = x_2$ and  $f(z_0) = z_1$ and $g(z_0) =z_2$.  Since $g$ is order preserving, we have $x_2 \le_2 z_2 \le y_2$.\\	
	\item There is $z_3 \in P_3$ with $f'(x_2) \le_3 z_3 \le_3 f'(y_3)$, so that there are $ z_1 \le_1 y_1$ with
	      $g'(y_1) = f'(y_2)$  and $g'(z_1) = z_3$ and $x_2 \le z_2$ and $f'(z_2) =z_3$ analogously.	
\end{enumerate}
$f'$ is a strict order embedding:\\ 
Let be  $x_2, y_2 \in P_2$ and $f'(x_2) \le_3 z_3 \le_3 f'(y_2)$ given for $z_3 \in P_3$.
Either $z_3 \in f'(P_2)$ and hence there is $f'(z_2)=z_3$ with $x_2 \le_2 \le_2 y_2$ or $z_3 \not \in f'(P_2)$.
Then there are $x_1,y_1, z_1,z'_1 \in P_1$ with $g'(x_1) = f'(x_2)$ and $g'(y_1) = f'(y_2)$ and  $g'(z_1) = z_3 = g(z_1)$ and $x_1 \le z_1$ and $z'_1 \le_1 y_1$. Due to the pushout construction there are  $x_0,  y_0 \in P_0$ with  $f(x_0) = x_1$ and $g(x_1) = x_2$ and  $f(y_0) = y_1$ and $g(y_1) = y_2$. Since $f$ is a strict order embedding we have 
additionally, $z_0,z'_0$ with $f(z_0) = z_1$ and $f(z'_0) = z'_1$ and $x_0 \le z_0 \le z'_0 \le y_0$. Due to pushout construction $g(z_0) = g(z'_0)$ and as $g$ is order preserving we have $x_2 = g(x_0) \le _2 g(z_0) \le_2 g(y_0) = y_2$
with $f'(g(z_0)) = z_3$.
\end{proof}

Next we investigat pullbacks in $\cPosets$.
\begin{lemma}[Pullbacks in \cPosets]
Given $(P_1,\le_1) \nach{g} (P_0,\le_0) \von{f} (P_2,\le_2)$  then there is the pullback
	     $(P_1,\le_1) \von{f'} (P_3,\le_3)  \nach{g'} (P_2,\le_2)$.
				Moreover,  \M \ is stable under pullbacks.
\end{lemma}
\begin{proof}
There is the pullback   $P_1 \von{f'} P_3 \nach{g'} P_2$ of $P_1 \nach{g} P_0 \von{f} P_2$   in $\cSets$.
$(P_1,\le_1) \von{f'} (P_3,\le_3)  \nach{g'} (P_2,\le_2)$ with $x_3 \le y_3$ if and only if $f'(x_3) \le_1 f'(y_3)$ and $g'(x_3) \le_1 g'(y_3)$ is pullback in \cPosets. Obviously, $f'$ and $g'$ are order-preserving mappings.

 \M-morphisms are monomorphisms and hence are preserved by pullbacks. 
\end{proof}

\begin{theorem}[$\cPosets$  is $\M$-Adhesive HLR Category.]
  \label{l.madHLR.posets}
\end{theorem}

\begin{proof}~

	\begin{enumerate}
		\item The class $\M$ in $\cPosets$  is  PO-PB compatible, since
		\begin{itemize}
			\item pushouts along \M-morphisms exist and \M \ is stable under pushouts,
		  \item pullbacks along \M-morphisms exist and \M \ is stable under pullbacks and
		  \item obviously, \M \ contains all identities and is closed under composition.
	\end{itemize}
	\item  In $\cPosets$ pushouts along \M-morphisms are $\M$-VK squares:
	  Let be given as : a  pushout (1) with $m \in \M$  and some commutative cube (2) with (1) in
the bottom and back faces being pullbacks in $\cPosets$.

 $
   \xymatrix@=3mm{
   			&\\
                       A  \ar[rr]|{m\in\M} \ar[dd]|{f}    
	         && B                   \ar[dd]|{g}     \\ 
	            & (1)  \\
		       C  \ar[rr]|{n}
		    && D
		 }
     $ \hfill
     $
    \xymatrix@=3mm{
                     &&&    A'	 \ar[ddd]|<<<<<{a}    \ar[dlll]|{f'}    \ar[drr]|{m'}     
		         &&  (2) \\
		                C'      \ar[ddd]|{c}                    \ar[drr]|{n'}
		   &&&&& B'      \ar[ddd]|{b}     \ar[dlll]|<<<<<{g'}          
		           &&                                                                \\
		          && D'      \ar[ddd]|<<<<<{d}                                        \\
                        &&& A                       \ar[dlll]|>>>>>{f}     \ar[drr]|{m}           \\
		                C                                      \ar[drr]|{n}
		   &&&&& B                       \ar[dlll]|{g}                       
		           &&                                                              \\
		          && D
		   } 	        	 
  $	
 
 \begin{description}
	 \item[$\Rightarrow$:] Let the top be a pushout in $\cPosets$. Pullbacks preserve  \M-morphisms, so  $m'\in \M$ and  
	      hence the top square is a pushout in $\cSets$ as well. As $\cSets$ is adhesive, the front faces are pullbacks in $\cSets$ as well.
				Since the construction of pullbacks coincides in $\cSets$ and $\cPosets$, the front faces are pullbacks in $\cPosets$.
	
   \item[$\Leftarrow$:] Let the front faces be pullbacks in $\cPosets$, and hence pullbacks in $\cSets$. Since $m \in \M$
	      (1) is pushout in $\cSets$ as well. So, $\cSets$ being adhesive, we have the top square being a pushout 
				in $\cSets$. Moreover, $M'\in \M$ as the back face is a pullback  preserving \M-morphisms. So, the top
				is a pushout along $\M$ is $\cPosets$.
 \end{description}

\end{enumerate}
    Hence, by Def. \label{d.madHLR}  $(\cPosets,\M)$ is an $\M$-adhesive HLR-category.
\end{proof}

\begin{definition}
   The category of  place/transition nets with  transition priorities $\cPTp$  is given by $N=(P,(T,\le_T), pre, post, m_0)$
	with $pre, post: V(T,\le_T) \to P^\oplus$ and
	morphisms $f_P,f_T: N_1 \to N_2$ where $f_P$ is a mapping and $f_T$ is an order-preserving map.
	
	A transition $t\in T$ is enabled under a marking $m$, if $pre(t) \ge m $ and if for all $t’\in T$ being enabled under $m$ we have $t'\le_T t$.
\end{definition}

\begin{lemma}[$(\cPTp,\M)$ is an $\M$-adhesive HLR-category]
  \label{l.madHLr_PTp}
	 with $\M$ the net morphisms where $f_p$ is strict injective and $f_T$
	is a strict order embedding. 
  \end{lemma}
	
	\begin{proof}
The proof applies the construction for weak adhesive HLR categories (see Theorem 1 in \cite{PEL08}):\\
	 We know that  $(\cSets,\M)$  with $\M$ being the injective mappings is an $\M$-adhesive
HLR category and that $ (\_)^\oplus: \cSets \to \cSets$ preserves pullbacks along injective morphisms.
 As shown above $(\cPosets,\M)$  with $\M$ being the strict order embeddings is an $\M$-adhesive  
HLR category and that $V: \cPosets \to \cSets$ preserves pushouts along $\M$-morphisms.
So, the category $cPTp$ is isomorphic to the comma category $ComCat(V,(\_)^\oplus;I)$ with I = {1,2}, where $V: \cPosets \to \cSets$ is the forgetful functor from partial ordered sets to sets and $(\_)^\oplus$ is the free commutative
monoid functor and hence  an $\M$-adhesive.  
HLR category.
	\end{proof}

\begin{lemma}[$(\cdPTip,\M)$ is an $\M$-adhesive HLR-category]
  \label{l.madHLr_PTp}
	 with $\M$ the net morphims where $f_p$ is strict injective and $f_T$
	is a strict order embedding. 
  \end{lemma}
	\begin{proof}
	Similar to the proof of Lemma 1 in \cite{Pad12} using $\cPTp$ instead of $ \cPT$ as the basis.\end{proof}

%
%\section{Nested Application Conditions}
%\label{s.nestAC}
%\input{nestAC}

\section{Conclusion}
\label{s.conc}
%conc

The  tool \reconnet \ has been developed at the HAW Hamburg in various students projects. 
Up to now it supports the modelling and simulation of reconfigurable nets.  The nets and rules in Figs.  \ref{f.rules} and  \ref{f.trafo} have  been edited and computed by \reconnet. The tool's most important feature is the ability to create, modify and simulate reconfigurable nets through an intuitive graphic-based user interface (see \cite{EHOP12}). 

Ongoing work concern  the extension of the control structures. This includes the extension of rules with negative application conditions and an explicit representation of an abstract reachability graph based on \cite{Pad12}. 

\bibliographystyle{splncs}
\bibliography{bibDecoPriorInhib}

\end{document}